\newif\ifanon
\newif\ifsubmission
\newif\ifenhanced
\documentclass[runningheads]{llncs}

\makeatother
\usepackage[hidelinks, bookmarksopen=true]{hyperref}
\usepackage{bookmark}

\usepackage{preamble}

\title{Slow and Steady: Preventing MEV with Verifiable Delays}
\date{\today}

\begin{document}

\ifanon
    \author{ Anonymous }
\else
    \author{
        Zeta Avarikioti\inst{1,2}
        \and
        Dimitris Karakostas\inst{2}
        \and\\
        Karl Kreder\inst{3}
        \and
        Shreekara Shastry\inst{3}
    }
    \institute{
        TU Wien \and Common Prefix \and Dominant Strategies
    }
\fi

\maketitle

\begin{abstract}
    Our work presents a defense mechanism against Maximal Extractable Value
    (MEV) opportunities in distributed ledgers. The mechanism relies on the idea
    of enforcing a verifiable delay when generating transactions, such that a
    block creator cannot react to the appearance of a MEV opportunity
    without breaking liveness.
    \ifenhanced
    We then enhance our mechanism with Proof-of-Work (PoW) to defend
    predictable MEV opportunities. For both mechanisms, we 
    \else
    We
    \fi
    present positive results both in the Byzantine setting and in a game
    theoretic model of rational participants. We additionally present negative
    bounds that outline the limitations of this line of defense. Finally, we
    explore real-world implementation details of verifiable delays and show
    that, based on historical MEV data, our mechanism could realistically help
    prevent most existing MEV threats.
\end{abstract}

\section{Introduction}\label{sec:introduction}

A Hare was making fun of the Tortoise one day for being so
slow~\cite{aesop-fable}. In the end, the Tortoise beat the Hare.

In a fast-paced, high-frequency environment, slowing down merits
extra consideration, as it can sometimes yield surprising results.
The blockchain-based decentralized finance~\cite{werner2022sok} ecosystem is
such a fast-paced environment. Built on the premise that financial transactions
are bundled in blocks to form an immutable chain, the ecosystem that started
with Bitcoin~\cite{nakamoto2008bitcoin} was soon boosted by the introduction of
autonomous programs called smart contracts~\cite{buterin2013ethereum}. These
programs enabled the existence of programmable money~\cite{lee2021programmable},
which allows digital money to have automated behavior through computer programs, sometimes completely
autonomous.

However, this new paradigm brought along new threats. One
prominent family of hazards is Maximal Extractable Value (MEV)
opportunities (also
called Miner/Blockchain Extractable Value opportunities)~\cite{daian2020flash,obadia2021unity,qin2022quantifying}. MEV opportunities
refer to the ability of a block creator to inject or reorganize
transactions to maximize their profit. For instance, a malicious block creator
can front/back-run a user or ``sandwich'' a user's transaction between two
transactions of their
own. In this case, the execution of the user's transaction, within the context
of a blockchain-hosted application (\eg a smart contract), results in increased
profits for the block creator at the user's and/or application developer's expense.

MEV opportunities in distributed ledgers are problematic
mainly in two aspects.
First, assets flow from users to block creators. When the latter
have control
over the execution of a user's transaction and can potentially exploit it, the
system's usability is hurt.
Second, MEV opportunities pose an offchain incentive to block creators, which
cannot be controlled by the ledger protocol. Consequently, if such incentives
are high enough, profit-driven parties may divert from the protocol
in order to claim an MEV opportunity. 
Such behavior could jeopardize the system's security, potentially causing a failure of the consensus mechanism.

In this work, we explore a novel approach for countering the effects of MEV
opportunities. Inspired by the fabled Tortoise, we use verifiable delays to
slightly slow down the execution of a distributed ledger, in order to make it
more robust and stable by eliminating MEV opportunities.

\subsection*{Our Contributions and Roadmap}

Our main idea is preventing a party from quickly responding to a
transaction that is broadcast to the network. We achieve this by requiring
each transaction to be evaluated by a Verifiable Delay Function
(VDF)~\cite{boneh2018verifiable}, thus enforcing a delay between the time a user
selects a transaction's payload and when a transaction object, which is valid
for submission to the protocol's participants, is constructed.

In more detail,~\cref{sec:model} outlines the model within which our analysis
takes place and briefly reviews relevant preliminaries.

\cref{sec:protocol} describes our main defense mechanism. The defense is
implemented via a VDF-based transformation of a ledger's transaction validity
predicate. In our analysis, we define the notion of blockchain input causality,
where an adversary cannot create a valid transaction that meaningfully depends
on an honest transaction before the latter is finalized. Our analysis first
assumes the existence of a distributed ledger protocol that, in the absence of
MEV opportunities, guarantees liveness and is an equilibrium and is compliant
\wrt censorship. Then, we show that any such protocol can be transformed to a
new protocol which, when MEV opportunities do exist, guarantees blockchain input
causality and is also an equilibrium and compliant \wrt censorship. We also
present a negative game theoretic result, by describing an equilibrium where all
parties deviate from the protocol in a specific manner that tries to claim an
MEV opportunity.

\ifenhanced
\cref{sec:protocol-enhanced} considers a more elaborate setting, where the
adversary can predict the existence of a future MEV opportunity with some probability. We again present a transformation,
which makes use of both a VDF and Proof-of-Work (PoW). Using PoW enables us to
present a positive result in this enhanced setting, where a transformed protocol
is an equilibrium as long as the expected MEV opportunity is lower than the PoW
cost of constructing a transaction.
\fi

Finally,~\cref{sec:implementation} discusses various implementation details. We
review existing VDF proposals and evaluate their suitability in our application,
in terms of time and space efficiency. Then, we explore historic MEV
data. These data show that the overwhelming majority of MEV opportunities is in
the order of transaction fees, so our proposals could be realistically
implemented and protect most MEV cases, thus making distributed ledgers more
robust and user-friendly.

% \ifsubmission
% We note that, due to space constraints, the proofs of our results are included
% in~\cref{sec:proofs}, along with extra preliminaries in~\cref{sec:preliminaries}.
% \fi

\subsection*{Related Work}

MEV opportunities~\cite{daian2020flash,obadia2021unity,qin2022quantifying} have
been extensively studied in recent years, with various proposed
countermeasures~\cite{DBLP:conf/aft/HeimbachW22,DBLP:journals/corr/abs-2212-05111,DBLP:conf/defi/Yang0HCYZ24}
broadly falling in three categories.

First are MEV auction platforms~\cite{flashbots,mev-boost,bloxroute}, which act
as intermediaries between users (who offer transactions) and miners (who offer
block space). Particularly in Ethereum, MEV auctions are natively supported
via Proposer-Builder Separation (PBS)~\cite{buterin2022state}.
However, PBS requires trust among participants, since users
need to trust the block builders to guarantee their transactions' privacy and
safeguard them against MEV attacks, while miners/validators need to trust
the block creators and relayers to craft blocks correctly. In comparison, our
proposal requires no trusted third party and relies only on the 
ledger's security.

Second are time-based ordering properties, which rely on the notion of
transaction order
fairness~\cite{DBLP:conf/crypto/Kelkar0GJ20,DBLP:conf/fc/CachinMSZ22,DBLP:journals/iacr/KiayiasLS23,DBLP:conf/asiapkc/KelkarDK22,DBLP:conf/icbc2/CachinM24,DBLP:journals/corr/abs-2411-09981}.
These solutions prevent order manipulation by explicitly defining the properties
that transaction ordering must satisfy. These properties are receive-order
fairness~\cite{DBLP:conf/crypto/Kelkar0GJ20,DBLP:conf/asiapkc/KelkarDK22,juels2020fair}
where transactions are published in a first-come-first-serve manner, or
relative-order fairness~\cite{zhang2020byzantine,kursawe2020wendy} where the
relative order of two transactions should be guaranteed.
Our idea relies on the observation that relative order-fairness can be guaranteed
directly from liveness, for any two transactions that are submitted with a large
enough delay between them. Therefore, instead of trying to guarantee
order fairness in general, which is particularly tricky, we focus on
increasing the delay between a transaction that offers a MEV opportunity and any
transaction that tries to attack it.

Third, there exist content-agnostic ordering proposals. Here,
transactions are first committed on-chain and their ordering is
decided at a later point in time, when the block creators can no longer manipulate it.
This is achieved by encrypting a transaction's content upon its
publication~\cite{asayag2018fair}, implementing
a randomness beacon that decides the ordering of past
transactions via commit-reveal schemes~\cite{asayag2018fair,DBLP:conf/opodis/AlposACY23}, distributed key generation~\cite{DBLP:journals/iacr/KavousiLJD23}, or relying on
trusted hardware~\cite{cryptoeprint:2017/1153}. A practical drawback of these
solutions is that block creators do not know a priori which transactions will be
eventually published in the final ledger and in which order. As a result, they
cannot compute their expected rewards, \ie their utility, in order to make an
informed, rational choice on how to participate. Additionally, the computation
of the ledger's state, \eg account balances or smart contract state, is
necessarily delayed until transaction ordering is eventually decided. 
Furthermore, securely implementing a public randomness beacon is highly
non-trivial and, although various protocols
exist~\cite{DBLP:conf/eurosp/KavousiWJ24}, if implemented incorrectly may
result in bias hazards.
Our mechanism does not directly hide a transaction's content when block
creators receive it, thus avoiding the drawbacks of such solutions. However,
similar to these works, our solution increases the time needed to finalize
transactions, albeit the increase in our case corresponds to local computations
(that is the evaluation of a VDF by the user), instead of being part of the
distributed ledger protocol itself. Furthermore, our mechanism can be
implemented on an ad hoc basis, \eg on the smart contract level, and does not
need to be incorporated in the distributed ledger protocol.

\section{Model}\label{sec:model}

Our analysis assumes a multi-party setting, following Canetti's formulation of the 
``real world''~\cite{cryptoeprint:2000/067}. An ``environment'' program $\env$ drives
the execution of a protocol, by spawning for each party $\party$ an instance
of an ``interactive Turing machine'' (ITM) which executes the protocol. Interaction between
parties is controlled by a program $C$, \st $(\env, C)$ create a system of ITMs.
We assume ``locally polynomially bounded'' systems to ensure polynomial time
execution.
% We also require sequential execution of parties, \ie first activating the
% adversary $\adversary$ and then the other parties in order.

% \noindent\emph{Notation}
% In the rest of the paper we will use the following notation:
% \begin{itemize}
%     \item $a.b$: the element $b$ of a complex object $a$
%     \item $a_{:i}, a_{-i:}$: the first (and last resp.) $i$ bits of a number $a$
%     \item $|\block', \block|$: the distance between blocks $\block'$ and $\block$ of a chain,
%         \ie the number of blocks in the sub-chain starting from $\block'$ and ending at $\block$;
%         in some cases we will somewhat overload the notation, \st the values $\block', \block$ may be blocks or references to blocks (that is their hash)
%     \item $A \concat A'$: the concatenation of $A$ and $A'$, which might be blocks,
%         chains, or (ledgers of) transactions
%     \item $\chain[i]$: the $i$-th block of a chain $\chain$; if $i=-1$ then it denotes the 
%         last block (also called head or tip) 
%     \item $\chain[:i], \chain[-i:]$: the first (and last resp.) $i$ blocks of a chain $\chain$
%     \item $\ledger \prec \ledger'$: the prefix operation of ledgers, \ie $\ledger$ is a prefix of $\ledger'$
%     \item $\block_{\tx}$: the block $\block$ in which a transaction $\tx$ is published
% \end{itemize}

\emph{Network.}
The execution proceeds in (a polynomial number of) rounds of size $\Delta$.
In each round $\round$, each party is activated and performs various operations based
on its algorithm. A message honestly produced at round $\round$ is delivered to
all other parties at the beginning of round $\round{+}1$, \ie we assume
a \emph{synchronous network}.
Finally, we assume a diffuse functionality, \ie a gossip protocol which
allows the parties to share messages without the need of a fully connected graph.

\emph{Parties.}
The set of $\totalParties$ parties $\partySet$ remains fixed for the duration of
the execution. Each party $\party$ holds a (fixed) percentage of power
$\power_\party$ which enables its participation in the distributed ledger protocol.
For example, if a Proof-of-Work (PoW) mechanism is used, the power corresponds
to hashing power, whereas in Proof-of-Stake (PoS) it corresponds to stake.

% \ifsubmission
% Our work also makes use of two known primitives, Verifiable Delay Functions
% (VDFs) and commitments (cf. \cref{sec:vdf-definition} and \cref{sec:commitment}).
% \fi

\subsection{Cryptographic Model}

In the cryptographic model, some parties are assumed honest, \ie follow the
protocol, whereas the rest are Byzantine, \ie may arbitrarily deviate from
the protocol specification. Note that incentives do not play a role in the
analysis under this model; they are instead considered in the game theoretic
analysis.

\emph{Adversary.}
The adversary $\adversary$ is an ITM which, upon activation, may ``corrupt'' a number of
parties by sending a relevant message to the controller after $\env$ instructs it to do so.
As a result, when a corrupted party is supposed to be activated, $\adversary$ is activated
instead.
$\adversary$ controls $\adversarialParties$ parties, out of the total $\totalParties$,
and its power is the sum of the corrupted parties' power.
% Therefore, the power of the adversary $\power_\adversary$ is equal to the sum of all corrupted parties.
Furthermore, $\adversary$ is ``adaptive'', \ie corrupts parties on the fly,
and ``rushing'', \ie decides its strategy and acts after observing the other
parties' messages.

\subsection{Game Theoretic Model}

In the game-theoretic model, we assume that all parties are rational.
Each party $\party$ employs a strategy $\strategy$, that is a set of rules
and actions that the party makes depending on its input. In other words,
$\strategy$ defines the algorithm that $\party$ runs and its
actions within the distributed protocol. 

Each party has a well-defined utility $\utility$ and chooses a strategy $\strategy$
with the goal of maximizing their utility.
The literature of game theoretic analyses of blockchain protocols considers various
types of utilities~\cite{DBLP:journals/corr/abs-1905-08595}:
\begin{inparaenum}[(i)]
    \item absolute rewards;
    \item absolute profit, \ie absolute rewards minus cost;
    \item relative rewards, \ie the percentage of the party's reward among
        all allocated rewards;
    \item relative profit, \ie relative rewards minus (absolute) cost;
    \item deposits and penalties, \eg using a slashing mechanism;
    \item external rewards, \eg bribes.
\end{inparaenum}
Our analysis will consider two utilities. First,~\cref{sec:protocol}
considers absolute rewards that consist of three types of income: fixed
block rewards, transaction fees, and MEV opportunities.
\ifenhanced
Next,~\cref{sec:protocol-enhanced} considers absolute profit, where the income
is computed as before and cost corresponds to Proof-of-Work computations.
\fi

We will use two game theoretic tools, Nash equilibrium and compliance.
% \ifsubmission
% Due to space constraints, we refer to Appendices~\ref{sec:nash-equilibrium-definition}
% and~\ref{sec:compliance-definition} for a detailed description of these
% primitives.
% \fi

% \ifsubmission
% \subsection{Approximate Nash Equilibrium}\label{sec:nash-equilibrium-definition}
% \else
\subsubsection{Approximate Nash Equilibrium.}\label{sec:nash-equilibrium-definition}
% \fi
We assume a set of parties $\party_1, \dots, \party_n$. Each party $\party_i$
employs a strategy $\strategy_i$, that defines a set of rules and actions that
the party follows depending on how the execution progresses. A strategy profile
$\strategyProfile$ is a vector of all parties' strategies. Intuitively,
$\strategyProfile$ it is an $\epsilon$-Nash equilibrium if no party can increase
its utility more than $\epsilon$ by unilaterally deviating from the profile
(Definition~\ref{def:nash-equilibrium}).

\begin{definition}[$\epsilon$-Nash equilibrium]\label{def:nash-equilibrium}
    Let:
    \begin{itemize}
        \item $\epsilon$ be a non-negative real number;
        \item $\partySet$ be the set of all parties;
        \item $\strategySet$ be the set of strategies that a party may employ;
        \item $\strategyProfile = \langle \strategy_i, \strategy_{-i} \rangle$
            be a strategy profile, where $\strategy_i$ is the strategy followed by $\party_i$
            and $\strategy_{-i}$ are the strategies employed by all parties except $\party_i$;
        \item $\utility_i(\strategyProfile)$ be the utility of party $\party_i$ under a strategy
            profile $\strategyProfile$.
    \end{itemize}
    $\strategyProfile$ is an $\epsilon$-Nash equilibrium \wrt a utility vector
    $\bar{\utility} = \langle \utility_1, \dots, \utility_{|\partySet|} \rangle$ if:\\
    $\forall \party_i \in \partySet \; \forall \strategy'_i \in \strategySet:
    \utility_i(\langle \strategy_i, \strategy_{-i} \rangle) \geq \utility_i(\langle \strategy'_i, \strategy_{-i} \rangle) - \epsilon$.
\end{definition}

% \ifsubmission
% \subsection{Compliance}\label{sec:compliance-definition}
% \else
\subsubsection{Compliance}\label{sec:compliance-definition}
% \fi
Compliance~\cite{DBLP:conf/aft/KarakostasK022} builds on
the notion of an infraction predicate $\infractionPredicate$. Such predicate
abstracts a specific deviant behavior that the analysis aims to capture.
Intuitively, a protocol is compliant \wrt a specific behavior if, starting from
the setting where all parties follow the protocol, no party will be incentivized
perform the behavior under question. For example, if a protocol is compliant
\wrt to abstaining, no party will be incentivized to abstain, even if the
protocol is not an equilibrium due to the existence of other, profitable deviations.

More precisely, a strategy $\strategy$ is $\infractionPredicate$-compliant if, for
every execution when a party follows $\strategy$, it holds that
$\infractionPredicate=0$. Using this notion, a protocol $\proto$ can be defined
as $(\epsilon, \infractionPredicate)$-compliant. Specifically, consider the
profile $\strategyProfile$ where all parties follow $\proto$. Starting from
$\strategyProfile$, a new profile $\strategyProfile'$ is reachable if a party
can increase its utility by $\epsilon$ when deviating from $\strategyProfile$.
Henceforth, a set of strategies is formed, which includes all strategies that are
reachable from $\strategyProfile$ via paths of such stepwise deviations. In this
setting, $\proto$ is compliant \wrt $\infractionPredicate$, if this set contains
only profiles consisting of $\infractionPredicate$-compliant strategies. In
other words, if a protocol is compliant, parties may deviate from it, but only
in ways that do not violate $\infractionPredicate$.
% \za{needs some intuition, hard to follow}

\subsection{Distributed Ledger}\label{sec:ledger-definition}

A distributed ledger protocol $\proto$ should guarantee two fundamental
properties: safety (Definition~\ref{def:safety}) and liveness
(Definition~\ref{def:liveness}). Here, an ``honest'' party is a party that
follows $\proto$ and $\ledger$ denotes a ledger of transactions, with
$\ledger_{\party, \round}$ being the ledger that a party $\party$ outputs on
round $\round$. Also, $\prec$ denotes the prefix operations, \st $A \prec B$
means that $A$ is a prefix of $B$.

\begin{definition}[Safety]\label{def:safety}
    A distributed ledger protocol is safe if, for any honest parties $\party, \party'$ 
    and any rounds $\round, \round'$, it holds that
    $(\ledger_{\party, \round} \prec \ledger_{\party', \round'}) \cup
    (\ledger_{\party', \round'} \prec \ledger_{\party, \round})$.
\end{definition}

\begin{definition}[Liveness]\label{def:liveness}
    Assume that all honest parties receive a transaction $\tx$ at round
    $\round$. A distributed ledger protocol is live, with parameter
    $\livenessParam$, if for every honest party $\party$ it holds that $\tx \in
    \ledger_{\party, \round+\livenessParam}$. % \za{what does it mean belong? the definitions are not well defined stand alone i think}
\end{definition}

Transaction $\tx$ is \emph{finalized} if, on round $\round$,
$\tx \in \ledger_{\party, \round}$ for some honest party $\party$.

In this work, we consider blockchain-based distributed ledger
protocols, where transactions are published in blocks. Each
block $\block$ contains a reference to exactly one other block, so the
blocks form a directed acyclic graph, the root of which is a common reference
``genesis'' block $\genesis$. The ledger is constructed by picking one of the
graph's branches, that is a chain of blocks that starts from
$\genesis$.\footnote{Usually, the chain decision rule uses some notion of
weight, \st the longest or ``heaviest'' chain is chosen. Our work only assumes a
deterministic algorithm that picks a single chain, without any restrictions on
how this algorithm operates.}

We also require that the considered protocols have an important
``proportionality'' property. Specifically, on each round, the expectation that
a party $\party$ produces a block, which extends the chain of any party that
follows $\proto$, is at most $\power_\party + \negl$.\footnote{$\negl$ is a
negligible function in the security parameter $\secparam$. A function $f:
\mathbb{N} \rightarrow \mathbb{R}$ is negligible if for every $c \in \mathbb{N}$
exists $\secparam_c \in \mathbb{N}$, \st $\forall \secparam > \secparam_c :
|f(\secparam)| < \secparam^{-c}$.}
This property is approximately\footnote{Typically, the random variable
satisfies the Chernoff concentration bounds. We use this stricter
definition for ease of analysis, without affecting the results qualitatively.}
present in ledger protocols that rely on
hardware Sybil resilience assumptions, such as
PoW~\cite{dwork1992pricing,jakobsson1999proofs} or
Proof-of-Space~\cite{dziembowski2015proofs}, as well as -- to some extent -- PoS
protocols.\footnote{Many PoS protocols are
predictable~\cite{bagaria2022proof}, \st they fix the near-future leader schedule 
ahead of time. Although in the long term, the probability
that any party can produce a block at a given round depends only on its power,
in the short term a party can predict with certainty the round when they will be
able to produce a block. This is a distinction from typical PoW-based protocols,
where no party can predict who will create a block in the following round with
a probability higher than their power.}

\subsection{Verifiable Delay Functions}\label{sec:vdf-definition}

A Verifiable Delay Function (VDF)~\cite{boneh2018verifiable} is a primitive that consists of a
triple of algorithms $\langle \setup, \eval, \verify \rangle$.
$\setup(\secparam, \vdfParam)$ takes a security parameter $\secparam$ 
and a delay parameter $\vdfParam$ and
outputs the VDFs public parameters $\msf{pp}$. $\eval(\msf{pp}, \payload)$ takes
an input $x$ from the VDF's domain and outputs a value $\vdfVal$ in the
function's range and, optionally, a short proof $\vdfProof$.\footnote{The short
proof is not a necessary element of the VDF, but it is used in practice to
speed-up the verification process.}
Finally, $\verify(\msf{pp}, x, \vdfVal, \vdfProof)$\footnote{For notation
simplicity we will omit $\msf{pp}$ from $\eval$ and $\verify$ in the rest of the
paper.} verifies that $\vdfVal$ and $\vdfProof$ are correct VDF output and proof
resp.\ on $x$.

A VDF should satisfy the following informal properties:
\begin{itemize}
    \item \emph{sequential}: honest parties can compute 
        $(\vdfVal, \vdfProof) \leftarrow \eval(x)$ in $\vdfParam$ sequential steps,
        while no adversary with a polynomial number of 
        parallel processors can distinguish $\vdfVal$ from random in significantly 
        fewer steps than $\vdfParam$.\footnote{Typically, ``significantly fewer''
        means that a proof of correctness for computation of length
        $t$ is computable in parallel to the computation with
        $\text{polylog}(t)$ processors~\cite{boneh2018verifiable}.}
    \item \emph{efficiently verifiable}: $\verify$ should be as fast as possible, 
        with an upper bound of $O(\text{polylog}(\vdfParam))$.
    \item \emph{unique}: for all inputs $x$, it is difficult to find $\vdfVal$ \st
        $\verify(x, \vdfVal, \vdfProof) = \text{True}$ and $\vdfVal \neq \eval(x)$.
\end{itemize}

\remark{Both the distributed ledger protocol and the VDF progress in sequential
time steps. Specifically, the ledger protocol progresses in rounds, whereas the
VDF's sequential step is (implicitly) a PRAM computation. Notably, these two
types of time steps are qualitatively different, since a round is orders of
magnitude larger than a PRAM computation. To enable meaningful comparisons (\eg
in \cref{thm:input-causality}), in the rest of this paper we assume a function
that transforms PRAM computations to ledger rounds. Therefore, $\vdfParam$
denotes the transformation of the VDF parameter (expressed in PRAM
computations) to ledger rounds.}

\ifsubmission
\subsection{Commitment Schemes}\label{sec:commitment}
\else
\subsubsection{Commitment Schemes}\label{sec:commitment}
\fi
A commitment scheme consists of two algorithms,
$\commit: \{ 0, 1 \}^\star \rightarrow \{ 0, 1 \}^\secparam$
and $\reveal: \{ 0, 1 \}^\secparam \times \{ 0, 1 \}^\star \rightarrow \{ 0, 1 \}$.
Intuitively, a commitment scheme enables a party to commit to a chosen value
and reveal it later, while keeping it hidden (hiding property) and not being
able to repudiate it (binding property). Note that, in order for the properties
to hold, the committed value should be large and random enough, so typically
the plain value is committed along with a large enough random nonce.  Briefly,
a cryptographic commitment scheme should satisfy the following properties:
\begin{itemize}
    \item \emph{Binding}: For all PPT algorithms, it should be infeasible to
        output $x \neq x'$ and $\msf{nonce}, \msf{nonce}'$, where
        $|\msf{nonce}| = |\msf{nonce}'| = 2^\secparam$, such that
        $\commit(\langle x, \msf{nonce} \rangle) = \commit(\langle x', \msf{nonce}' \rangle)$.
    \item \emph{Hiding}: Let $U_\secparam$ be the uniform distribution over the
        $2^\secparam$ opening values for security parameter $\secparam$. It
        should hold that, for all $x \neq x'$ the probability ensembles
        $\{ \commit(\langle x, U_\secparam \rangle) \}$ and $\{ \commit(\langle x', U_\secparam \rangle) \}$
        are computationally indistinguishable.
\end{itemize}

\section{MEV Protection via VDFs}\label{sec:protocol}

% \subsection{VDF-based Transformation}

Our transformation, that makes a distributed ledger protocol MEV resilient,
relies on a simple intuition. By requiring every transaction to be created
over a certain number of rounds, we enforce a delay on an MEV
attacker, between the moment they receive a transaction with an MEV
opportunity and the moment when they can publish a front-running transaction of
their own. If the required amount of rounds is high enough (more than
the ledger's liveness parameter), the honest transaction is finalized before
an adversarial transaction is created.

More precisely, we assume:
\begin{inparaenum}[(i)]
    \item a distributed ledger protocol $\proto$ (cf.~\cref{sec:ledger-definition});
    \item a $\vdf{=}\langle \setup, \eval, \verify \rangle$ with delay parameter
        $\vdfParam$ (cf.~\cref{sec:vdf-definition});
    \item a commitment scheme $\msf{CS} = \langle \commit, \reveal \rangle$.
\end{inparaenum}

As is standard in distributed ledger literature, $\proto$ defines a validity
predicate $\validate$ which, given a ledger and a transaction, outputs
$\text{True}$ if the transaction is allowed to be appended to the ledger. In
essence, this predicate implements the logic of the ledger's application, for
example preventing double-spending or ensuring no new assets are created
incorrectly.

The first step of our mechanism is to transform the transactions in order
to enforce a verifiable delay. Specifically, a transformed transaction
is a tuple
$
\tx = \langle \payload, (\vdfVal, \vdfProof) \rangle
$,
where $\payload$ is the transaction submitted to $\proto$ and $\vdfProof$ is the
proof output by the employed VDF.

Second, our mechanism defines a new validity predicate $\validate_\vdf$, which is a
transformation of $\validate$ and is parameterized by the VDF.
A transaction $\tx = \langle \payload, (\vdfVal, \vdfProof) \rangle$ is now
valid \wrt a ledger $\ledger$ if:
\begin{equation}\label{eq:validity-predicate}
    \begin{split}
        \validate_{\vdf, \msf{CS}}(\tx, \ledger) = \\
        \validate(\payload, \ledger) \land \\
        \verify(C, \vdfVal, \vdfProof) \land \reveal(C, \payload) = 1
    \end{split}
\end{equation}

\ifsubmission
\else
Note that in all of our work we assume that a transaction's payload $\payload$
is an element in the domain of $\commit$.
\fi

In short, our transformation adds one extra requirement for a transaction to be
valid: a VDF should have been evaluated on its commitment.\footnote{We use the
transaction's commitment, instead of the payload itself, to enable outsourcing
of the VDF computation, as will be discussed in
\cref{sec:implementation}. For this we assume that $\payload$ is large
enough s.t. no party can feasibly predict it using brute force.}

\remark{The sequentiality property of VDFs is defined over random domain
elements. In our case, the domain element is the output of the commitment
function. Therefore, for the property to translate in our setting, we assume
that the commitment function is modeled as a random oracle.}

\subsection*{Cryptographic Analysis}

To analyze the transformation in the Byzantine setting, the first step is to
define MEV resilience. Intuitively, a protocol is MEV resilient if an adversary
cannot front-run an honest transaction. In other words, for a transaction $\tx$
that is broadcast on round $\round$, an adversary cannot create a transaction
$\tx'$ which depends on $\tx$ and which is published ahead of $\tx$ in the final
ledger.

Fortunately, this notion already exists in the literature as ``input
causality''~\cite{reiter1994securely,cachin2001secure}. Briefly, input causality
ensures that, given a payload message $m$ sent by an honest client but not yet
delivered, the adversary cannot ask the system to deliver a message that depends
in a meaningful way on $m$~\cite{cachin2001secure}.

\cref{def:input-causality} adapts input causality in our setting. Here, a
``payload message'' is a transaction $\tx$. A message is ``sent'' when the
honest client broadcasts $\tx$ and it is ``delivered'' when $\tx$ is finalized
on the ledger. In this context, input causality is guaranteed if an adversary
cannot create a transaction $\tx'$, which passes the validity checks and which
meaningfully relies on $\tx$, before $\tx$ is finalized.

\begin{definition}[Blockchain Input Causality]\label{def:input-causality}
    Let transaction $\tx$ which is received by honest parties on round $\round$.
    No valid adversarial transaction $\tx'$, which meaningfully relies on $\tx$,
    can be received by an honest party $\party$ on round $\round'$ where $\tx
    \not \in \ledger_{\party, \round'}$, \ie before $\tx$ is finalized.
\end{definition}

Because the change we introduce is on the application level, the new ledger inherits
the security guarantees of the ledger protocol on which it is built.

\cref{thm:input-causality} %\ifsubmission (proof in~\cref{proof:vdf-input-causality}) \fi
shows that MEV resilience can be achieved if $\vdfParam \geq \livenessParam$. 
Intuitively, an honest transaction created on round $\round$ is finalized at
round $\round{+}\livenessParam$ (due to the ledger's liveness guarantee), which
is the earliest round on which an adversary can create a valid front-running
transaction (due to the imposed VDF delay).

\begin{theorem}[Blockchain Input Causality]\label{thm:input-causality}
    Assume:
    \begin{itemize}
        \item A distributed ledger protocol $\proto$
            (cf.~\cref{sec:ledger-definition}), which guarantees liveness with
            parameter $\livenessParam$ with negligible error probability.
        \item A VDF function $\vdf$ with delay parameter $\vdfParam$
            (cf.~\cref{sec:vdf-definition}). 
        \item A binding commitment scheme $\msf{CS} = \langle \commit, \reveal \rangle$.
    \end{itemize}
    The ledger protocol $\vdfProto$, which is constructed by applying the
    validity predicate transformation $\validate_{\vdf, \msf{CS}}$ (cf.
    Eq.~\eqref{eq:validity-predicate}) on $\proto$ with VDF parameter $\vdfParam
    > \livenessParam$, guarantees input causality (\cref{def:input-causality})
    with negligible error probability.
\end{theorem}
% \ifsubmission
% \else
\begin{proof}
    We will prove violating $\vdfProto$'s input causality implies
    violating $\proto$'s liveness.
    
    Let a transaction $\tx$ that is received on round $\round$ by honest parties
    that run $\vdfProto$. Any transaction $\tx'$ which meaningfully relies on
    $\tx$ is created on round $\round$ at the earliest, that is when $\tx$ is
    broadcast to all parties. Due to the second requirement of
    Eq.~\eqref{eq:validity-predicate}, $\tx'$ is valid on round $\round +
    \vdfParam$ at the earliest, which holds due to the sequentiality property of
    the VDF. If input causality is violated, then $\tx$ is not finalized until
    after round $\round + \vdfParam > \round + \livenessParam$, so liveness of
    $\vdfProto$ is violated.

    However, observe that the probability that a party creates a block at any
    round is the same in both $\proto$ and $\vdfProto$. This holds because this
    probability depends only on the party's power
    (cf.~\cref{sec:ledger-definition}) and not on the existence of a
    transaction. Therefore, liveness of a transaction $\tx$ does not depend on
    the existence or validity of another transaction $\tx'$. Since transaction
    validity is the only difference between $\proto$ and $\vdfProto$, if an
    adversary $\adversary$ can break liveness of $\vdfProto$, then $\adversary$
    can also break liveness of $\proto$.

    Finally, the commitment scheme's binding property guarantees that the
    VDF evaluation of adversarial transactions is done \wrt the
    adversarial payload.
    \qed
\end{proof}

% \fi

We note that the hiding property of the commitment scheme is not needed for the
security proof. Rather, it is useful in practice, in order to enable
outsourcing of the VDF computation to untrusted third parties (cf.
\cref{sec:implementation}).

\subsection*{Game Theoretic Analysis}

\emph{Utility.}
We assume that each party's utility is the absolute amount of income that is the sum of 
three types of rewards:
\begin{inparaenum}[(i)]
    \item fixed block rewards;
    \item transaction fees;
    \item MEV opportunities.
\end{inparaenum}
Note that we do not consider costs here.

\emph{Executions.}
In every execution, we assume that each transaction defines the following two
non-negative types of utility for protocol participants:
\begin{inparaenum}[(i)]
    \item an amount of fees $\txFee$;
    \item a MEV opportunity $\mevOpportunity$.
\end{inparaenum}
We define the following two types of executions:
\begin{inparaenum}[(i)]
    \item $\execution_\msf{+MEV}$: every transaction yields a MEV opportunity
        $0 \leq \mevOpportunity < \infty$;
    \item $\execution_\msf{-MEV}$: every transaction yields a MEV opportunity
        $\mevOpportunity = 0$.
\end{inparaenum}
Intuitively, in every execution $\execution_{+\text{MEV}}$ transactions might
offer MEV opportunities, whereas in $\execution_{-\text{MEV}}$ there exist no
MEV opportunities.

\emph{Reward distribution.}
We consider per-block reward mechanisms, where for each block in the final
ledger, the block's creator gets all three types of rewards associated with the
block, \eg Bitcoin's reward mechanism.

Theorem~\ref{thm:protocol-equilibrium} % \ifsubmission (proof in~\cref{proof:vdf-positive}) \fi
shows that, if a protocol is an equilibrium
\emph{in the absence} of MEV opportunities, then our VDF-based transformation
makes the (transformed) protocol an equilibrium when MEV opportunities do exist.
This is a very useful result for protocol designers, as it describes a clear way
for securing ledgers against MEV attacks in the rational
setting. That is, even if the vanilla protocol $\proto$ is not an equilibrium if
MEV opportunities exist, it can be directly transformed to a protocol that is an
equilibrium by using our proposal.

\begin{theorem}[Positive equilibrium]\label{thm:protocol-equilibrium}
    Assume:
        \begin{itemize}
            \item No party controls more than $\livenessThreshold$ power.
            \item A distributed ledger protocol $\proto$
                (cf.~\cref{sec:ledger-definition}), which guarantees liveness
                with parameter $\livenessParam$ with negligible error
                probability, if all parties that follow $\proto$ control $1 -
                \livenessThreshold$ power.
            \item A VDF function $\vdf$ with delay parameter $\vdfParam$
                (cf.~\cref{sec:vdf-definition}). 
            \item A binding commitment scheme $\msf{CS} = \langle \commit, \reveal \rangle$.
            \item A ledger protocol $\vdfProto$, constructed by applying the
                validity predicate transformation $\validate_{\vdf, \msf{CS}}$ (cf.
                Eq.~\eqref{eq:validity-predicate}) on $\proto$, with VDF parameter
                $\vdfParam > \livenessParam$.
        \end{itemize}
    If $\proto$ is an $\epsilon$-Nash equilibrium for every execution
    $\execution_\msf{-MEV}$, then $\vdfProto$ is an $\epsilon'$-Nash
    equilibrium for every execution $\execution_\msf{+MEV}$, where $\epsilon' =
    \epsilon + \negl$.
\end{theorem}
% \ifsubmission
% \else
\begin{proof}
    Let us assume a party $\party$ that deviates from the protocol $\vdfProto$, while
    all other parties follow $\vdfProto$. 

    First, observe that the only difference with $\vdfProto$ is the transaction
    validation predicate. Since a transaction's validity does not affect the
    expected amount of blocks that a party can produce, which only depends on
    its power (cf.~\cref{sec:ledger-definition}), $\party$ cannot increase their
    expected fixed block rewards by more than $\epsilon$ when deviating from
    $\vdfProto$.

    Second, in $\vdfProto$, $\party$ gets the fees of a transaction $\tx$ only if
    $\tx$ is included in a block of $\party$. However, similarly to the first argument,
    a transaction does not affect the parties' expected block production, so
    $\party$ cannot increase their expected transaction fee utility by more than
    $\epsilon$ when deviating from $\vdfProto$.

    Third, $\party$ can claim the MEV opportunity of $\tx$ only by censoring
    $\tx$ for at least $\vdfParam > \livenessParam$ rounds. However, by
    assumption this only happens with negligible probability, since $\party$
    controls less than $\livenessThreshold$ power. Therefore, $\party$ cannot
    increase their expected MEV utility by more than $\negl$ when deviating from
    $\vdfProto$.
\qed
\end{proof}

% \fi

\cref{thm:negative-equilibrium} % \ifsubmission (proof in~\cref{proof:vdf-negative}) \fi
presents another equilibrium where
the MEV opportunities are attacked, as opposed to
before. Intuitively, if all parties censor a transaction long
enough in order to claim its MEV opportunity, then no party has an incentive to
deviate and try to publish said transaction.

\begin{theorem}[Negative equilibrium]\label{thm:negative-equilibrium}
    Assume:
        \begin{itemize}
            \item No party controls more than $1 - \livenessThreshold$ power.
            \item A distributed ledger protocol $\proto$
                (cf.~\cref{sec:ledger-definition}), which guarantees liveness
                with parameter $\livenessParam$ with negligible error
                probability, if all parties that follow $\proto$ control $1 -
                \livenessThreshold$ power.
            \item A VDF function $\vdf$ with delay parameter $\vdfParam$
                (cf.~\cref{sec:vdf-definition}).
            \item A binding commitment scheme $\msf{CS} = \langle \commit, \reveal \rangle$.
            \item A ledger protocol $\vdfProto$, constructed by applying the
                validity predicate transformation $\validate_{\vdf, \msf{CS}}$ (cf.
                Eq.~\eqref{eq:validity-predicate}) on $\proto$, with VDF parameter
                $\vdfParam > \livenessParam$.
        \end{itemize}
    Let $\tx$ be a transaction that yields a MEV opportunity
    $\mevOpportunity_\tx > \epsilon$. We define the strategy
    $\strategy_\text{MEV}$, parameterized by $\tx$, where a party:
    (i) censors $\tx$ for $\vdfParam$ rounds;
    (ii) creates a transaction $\tx'$ that claims the MEV opportunity of $\tx$;
    (iii) tries to claim the MEV opportunity of $\tx$ when possible,
    that is on every round after $\tx'$ becomes valid, and while
    $\tx$ is not published in the ledger, the party tries to publish
    $\tx'$ before $\tx$;
    (iv) for every other operation, the party follows $\vdfProto$.
    
    If $\proto$ is an $\epsilon$-Nash equilibrium for every execution
    $\execution_\msf{-MEV}$, then $\strategy_\text{MEV}$ is an
    $\epsilon$-Nash equilibrium for every execution $\execution_\msf{+MEV}$.
\end{theorem}
% \ifsubmission
% \else
\begin{proof}
    We will follow a similar reasoning to Theorem~\ref{thm:protocol-equilibrium}.

    First, observe that trying to MEV-attack $\tx$ does not affect a party's
    expected block production which, by definition
    (cf.~\cref{sec:ledger-definition}), depends only on the party's power.
    Since $\proto$ is an $\epsilon$ equilibrium, no party $\party$
    can increase its utility in terms of fixed block rewards by more than
    $\epsilon$ when deviating from $\strategy_\text{MEV}$.

    Second, by assumption no party $\party$ can unilaterally enforce the
    publishing of $\tx$, since liveness can
    be violated by the rest of the parties who control more than
    $\livenessThreshold$ aggregate power. Therefore, if $\party$ deviates from
    $\strategy_\text{MEV}$ and tries to publish $\tx$ within $\vdfParam$ rounds
    of receiving it, all other parties will successfully censor $\tx$, meaning
    that $\party$ cannot claim the fees of $\tx$ in those rounds unless with
    negligible probability. Therefore, $\party$ can only claim the fees of $\tx$
    after all other parties stop trying to censor it, at which point all parties
    will follow $\vdfProto$, meaning that $\party$ cannot increase its utility
    in terms of transaction fees by more than $\epsilon$.

    Third, assume that $\party$ deviates from $\strategy_\text{MEV}$ by trying
    to publish $\tx$ within $\vdfParam$ rounds of receiving it. 
    If $\party$ succeeds, then its MEV utility \wrt $\tx$ would be $0$. If it
    doesn't succeed, then all other parties will try to
    claim the MEV opportunity after round $\round + \vdfParam$ (where $\round$
    is the round on which they received $\tx$). At this point though, the
    expectation that $\party$ claims the opportunity for themselves is at most
    equal to the same expectation when following $\strategy_\text{MEV}$.
\qed
\end{proof}

% \fi

% \dimitris{Stronger results: assume some parties are honest and all others are rational: is there an equilibrium and is it positive/negative?}

Finally, \cref{thm:censorship-compliance} % \ifsubmission (proof in~\cref{proof:censorship-compliance}) \fi
shows that if a protocol is compliant~\cite{DBLP:conf/aft/KarakostasK022} \wrt
censorship when no MEV opportunities exist, then the VDF-transformed
variant is compliant \wrt blockchain input causality when MEV opportunities do
exist.
We remind that compliance considers the case when starting from the setting
where everyone follows the protocol, parties may end up in arbitrary settings
via unilateral deviations. This result shows that, if censorship is not
performed when starting from the plain protocol under no MEV
opportunities, then blockchain input causality is achieved when starting from
the VDF variant even in the presence of MEV opportunities.

\begin{definition}[Censorship infraction predicate]\label{def:censorship-infraction}
    The predicate $\censorInfractionPredicate$ which captures censorship is 
    defined as follows. Let party $\party$ which, on round $\round$,
    receives a transaction $\tx$ that is valid \wrt the local chain that $\party$ adopts
    at the beginning of round $\round+1$. $\party$ performs censorship in an execution trace
    $\executionTrace$, \st $\censorInfractionPredicate(\executionTrace, \party) = 1$, if
    the block that $\party$ produces on round $\round+1$ does not include $\tx$.
\end{definition}

\begin{theorem}[Censorship compliance]\label{thm:censorship-compliance}
    Assume:
        \begin{itemize}
            \item No party controls more than $\livenessThreshold$ power.
            \item A distributed ledger protocol $\proto$
                (cf.~\cref{sec:ledger-definition}), which guarantees liveness
                with parameter $\livenessParam$ with negligible error
                probability, if all parties that follow $\proto$ control $1 -
                \livenessThreshold$ power.
            \item A VDF function $\vdf$ with delay parameter $\vdfParam$
                (cf.~\cref{sec:vdf-definition}).
            \item A binding commitment scheme $\msf{CS} = \langle \commit, \reveal \rangle$.
            \item A ledger protocol $\vdfProto$, constructed by applying the
                validity predicate transformation $\validate_{\vdf, \msf{CS}}$ (cf.
                Eq.~\eqref{eq:validity-predicate}) on $\proto$, with VDF parameter
                $\vdfParam > \livenessParam$.
        \end{itemize}

    If $\proto$ is $(\epsilon, \censorInfractionPredicate$)-compliant for every
    execution $\execution_\msf{-MEV}$, then $\vdfProto$ is $(\epsilon',
    \censorInfractionPredicate$)-compliant for every execution
    $\execution_\msf{+MEV}$, where $\epsilon' = \epsilon + \negl$.
\end{theorem}
% \ifsubmission
% \else
\begin{proof}
    Similar to the other theorems, $\party$ can only increase its utility when
    MEV opportunities exist by claiming them. Under $\vdfProto$, the only way to
    claim a transaction's MEV opportunity is to censor it for at least
    $\vdfParam$ rounds.

    The fact that $\proto$ is $(\epsilon, \censorInfractionPredicate$)-compliant
    means that $\party$ can increase its utility, that corresponds to block
    rewards and transaction fees, by less than $\epsilon$ when censoring a
    transaction. Note that MEV opportunities play no role here, since we assume
    that $\proto$ is executed under $\execution_\msf{-MEV}$.
    Therefore, since $\proto$ and $\vdfProto$ differ only in terms of
    transaction validity, when $\vdfProto$ is executed under
    $\execution_\msf{+MEV}$, $\party$ can again increase its utility, that
    corresponds to block rewards and transaction fees, by less than $\epsilon$
    when censoring a transaction. In other words, MEV opportunities cannot
    increase a party's expected block rewards and transaction fees.

    As a result, $\party$ can increase its utility by more than $\epsilon$ only
    if it can claim the MEV opportunity of a transaction $\tx$. To do this,
    $\tx$ should be censored for $\vdfParam$ rounds. Since $\vdfParam >
    \livenessParam$, claiming $\tx$'s MEV opportunity implies breaking liveness.
    By assumption $\party$ cannot unilaterally break liveness. Therefore, if all
    parties except $\party$ are $\censorInfractionPredicate$-compliant, then
    $\party$ cannot censor $\tx$ for $\vdfParam$ rounds unless with negligible
    probability (of breaking liveness). Hence $\party$'s expected utility in
    terms of $\tx$'s MEV opportunity is negligible.
\qed
\end{proof}

% \fi

\section{Implementation Details}\label{sec:implementation}

We now discuss details about a real-world implementation
of our mechanism.

\subsection{Verifiable Delay Functions}

\emph{VDF Delay Parameter.}
\cref{sec:protocol} shows that the VDF delay parameter should be larger
than the liveness parameter. In practice, the liveness parameter is
hard to estimate as it depends on various factors, such as network
congestion, connectivity delays, possible attacks, etc. Nonetheless, under normal
circumstances transactions are published in the next few blocks after being sent
to the network. For example, the median confirmation time for Bitcoin
transactions is $40{-}60$ minutes, whereas in Ethereum it is between $15{-}20$
minutes.\footnote{Confirmation time is the delay between sending
a transaction to the network and it being considered finalized on the ledger. As
of January $2025$, the delay until a Bitcoin transaction is published on a
block is $5{-}10$ minutes (cf.
\href{https://www.blockchain.com}{blockchain.com}). Also, cryptocurrency
exchanges, like
\href{https://support.kraken.com/hc/en-us/articles/203325283-Cryptocurrency-deposit-processing-times}{Kraken}
and \href{https://help.coinbase.com/en/coinbase/getting-started/crypto-education/glossary/confirmations}{Coinbase},
set the safety parameter to $2{-}3$ block confirmations, \ie $30{-}40$ minutes. Similarly, in
Ethereum the delay until a single confirmation ranges between a few seconds and
$5$ minutes, whereas the required time for enough confirmations to be produced
is approx. $15$ minutes.}

Taking a conservative approach, we recommend setting the VDF delay parameter to
be twice the safety parameter. Therefore, in Bitcoin this would be approx. $60$
minutes, whereas in Ethereum it would take approx. $30$ minutes.

\emph{VDF Candidates.}
Two primary VDF candidates can be used in practice:
Wesolowski's~\cite{cryptoeprint:2018/623} and
Pietrzak's~\cite{cryptoeprint:2018/627}. These have been explored and analyzed
in the
literature~\cite{cryptoeprint:2018/712,cryptoeprint:2022/755,cryptoeprint:2020/332,cryptoeprint:2023/626},
so here we will offer a brief overview.

Wesolowski's VDF~\cite{cryptoeprint:2018/623} has been previously analyzed
extensively. The analysis of~\cite{cryptoeprint:2020/332} shows that, for
$2^{25}$ exponentiations, the VDF evaluation and the computation of the
proof\footnote{The proof is used for speeding up the verification of the VDF
result.} requires approx. $95$ seconds, whereas verification takes approx. $0.3$
ms. This work also notes that the best available hardware (FPGA) for evaluating
this VDF is approximately $29$ faster than a CPU. Chia's VDF
contest~\cite{chia-vdf-contest} offers similar results. Here, for $2^{23}$
iterations, the VDF evaluation requires between $91.3{-}97.4$ seconds (with and
without SIMD/GPU resp.).

Pietrzak's VDF~\cite{cryptoeprint:2018/627} has been recently analyzed in the
context of Sprints~\cite{cryptoeprint:2023/626}, a consensus protocol that
combines PoW and this VDF. This analysis shows that, for a VDF evaluation that
corresponds to approx. $570$ seconds ($95$\% of a $10$ minute block time), the
corresponding verification requires $100{-}500$ms.

Our application uses the VDF only to impose a delay on transaction generation.
However, the VDF proofs need to be verified by all nodes in the network, since a
transaction's validity depends on them. To choose which option (Wesolowski's or
Pietrzak's) is suitable for our use case, we estimated the time needed to verify
a VDF proof compared to the duration of its computation.

\emph{Experimental evaluation.}
In our experiments, we evaluated the two VDF candidates on a Bitcoin-like
transaction object for various delay parameters. For each parameter, we repeated
the computation $50$ times and obtained the average time needed for the VDF
computation and for the verification of the VDF proof.

For Wesolowski's VDF, we used Harmony's
implementation.\footnote{\url{https://github.com/harmony-one/vdf}} We set the
difficulty to $5{,}000$, resulting in  $14.966$ seconds average computation
time, with average verification time $736$ ms. We then increased the difficulty
and obtained the following measurements:
\begin{inparaenum}[(i)]
    \item avg. computation time $151.965$ seconds, mean verification time $728$ ms;
    \item avg. computation time $597.247$ seconds (approx. $10$ mins), mean verification time $731$ ms;
    \item avg. computation time $1045.754$ seconds (approx. $17$ mins), mean verification time $762$ ms;
    \item avg. computation time $1801.844$ seconds (approx. $30$ mins), mean verification time $744$ ms;
    \item avg. computation time $3677.779$ seconds (approx. $60$ mins), mean verification time $782$ ms.
\end{inparaenum}
Evidently, the verification time remains less than $800$
milliseconds even when the VDF computation requires $60$ minutes. 
Additionally, the proof size is the same in all of these executions, consistently
below $516$ bytes; although the size increase is significant,\footnote{For
comparison, the average Bitcoin transaction is approx. $400$ bytes
[\url{https://bitcoinvisuals.com/chain-tx-size}, January $2025$].}, it is not
prohibitively expensive in practice. In summary, Wesolowski's  VDF is a
prime choice for our use case.

For Pietrzak's VDF, we used the implementation of Injective Labs.\footnote{\url{https://github.com/InjectiveLabs/vdf}}
We set the difficulty to
$12{,}000{,}000$, which resulted in an average computation time of $12.36$ seconds;
for these computations, the average verification time was $8.09$ seconds.
In essence, the verification time for Pietrzak's VDF is logarithmic to the
computation time, which makes it not suitable for our use case.

\emph{Outsourcing VDF computation.}
The VDF computation can be outsourced, \eg using protocols like
OpenSquare~\cite{DBLP:conf/ccs/ThyagarajanGBKS21}. This could enable a
competitive market where parties offer VDF evaluation as a service while using
state-of-the-art, efficient hardware, \st the computation is done efficiently
and cost-effectively. Note that, in our application, it is important to maintain
the assumption that a transaction's payload is received by any party, other than
the user that creates it, only \emph{after} the VDF evaluation has finished. Otherwise, \eg
if the user outsources the evaluation to an untrusted service by directly
submitting the payload, the service could start evaluating their own MEV-attacking
transaction in parallel. To avoid relying on trust, our
transformations require the VDF evaluation to be performed on the payload's commitment.
Therefore, a user can send the commitment to a service and, since the
commitment scheme is hiding, the service does not learn the payload while
evaluating the VDF.

\subsection{Historical MEV Opportunities}

As a final step, we will shortly review historical cases of MEV opportunities.
This will offer a better view of the severity of the threat in the ecosystem, as
well as inform our results. %, \eg the bounds of the enhanced transformation (\cref{sec:protocol-enhanced}).

First, a recent paper quantified MEV in Layer-2 protocols
(Polygon, Arbitrum, Optimism) and offers interesting
insights~\cite{bagourd2023quantifying}. In Polygon, the researchers
analyzed $7.7$M transactions, which yielded \$$213$M of total extracted MEV.
Although the top opportunity offered \$$165$M in profit, $99.99$\% of all
transactions corresponded to $\leq \text{\$}100$ MEV profit, of which $90$\%
corresponded to transactions with $\leq \text{\$}1$. In Arbitrum, $87{,}706$
transactions were analyzed, at a total extracted MEV \$$250$K. Among them,
$99.99$\% corresponded to $\leq \text{\$}100$ MEV profit, of which $90$\%
corresponded to transactions with $\leq \text{\$}10$ MEV. In Optimism, $900$K
MEV transactions were analyzed with total extracted MEV \$$120$K. Again,
$99.99$\% of all transactions corresponded to $\leq \text{\$}30$ MEV profit, of
which $90$\% corresponded to transactions with $\leq \text{\$}1$. As can be
seen from these results, although the maximum profit can be very high (in the
order of millions USD), the vast majority of opportunities are relatively small (in the order of transaction fees).

Second, we reviewed the publicly available data that MEV Boost
offers.\footnote{\url{https://mevboost.pics}}
This set contains over $3$M blocks with MEV
opportunities between $15{/}9{/}2022$-$18{/}10{/}2023$, that is a block with
opportunities every approx. $11$ seconds. Although the maximum extracted MEV was
\$$3{,}149{,}321$, the median MEV per block is \$$84.96$, with $90$\% of blocks
offering less than \$$345.34$ in MEV profit. Note here that multiple MEV
opportunities may exist per blocks, so the same numbers on the individual
opportunity level could be smaller. Nonetheless, these results are consistent
with those from~\cite{bagourd2023quantifying}.

Finally, we parsed Flashbots' publicly available dataset of MEV opportunities.\footnote{\url{https://collective.flashbots.net/t/publishing-flashbots-protect-and-mev-share-data}} This dataset
offers approx. $3.2$M cases of sandwich and over $2.8$M cases of arbitrage
opportunities. The data set spans $18$ months, between
$15{/}1{/}2022$-$20{/}6{/}2023$, where an opportunity emerged every $7$ seconds
on average. We sampled and analyzed $35{,}000$ cases from each type at random.
For sandwich cases, the maximum opportunity
was \$$3.3$M,\footnote{Transactions
\href{https://etherscan.io/tx/0x34ecd811699e410cefa1d92d73ec544f39a169ccc06c5104b7571fa31e2a6f16}{0x34ecd811699e410cefa1d92d73ec544f39a169ccc06c5104b7571fa31e2a6f16}
and \href{https://etherscan.io/tx/0x06d48c42751864bb4b3a733e880ba23a8623f839c8a5440da23c4e22363abe10}{0x06d48c42751864bb4b3a733e880ba23a8623f839c8a5440da23c4e22363abe10}.}
but the median MEV per case was \$$92$, with $90$\% of
cases offering less than \$$1{,}175$ in MEV profit. Arbitrage cases are slightly
higher, with the maximum opportunity being \$$9.18$M,\footnote{Transaction
\href{https://etherscan.io/tx/0x08126591af2acfe355e0668b6cc0c7b18257163c2dd9e067bb5842c68c879668}{0x08126591af2acfe355e0668b6cc0c7b18257163c2dd9e067bb5842c68c879668}.}
the median MEV per case at \$$318$, and $90$\% of cases offered less than
\$$4{,}165$. We also find that claiming these opportunities requires, on average,
approx. $4.78$ transfers for arbitrage and $6.06$ for sandwich cases.

In summary, multiple sources offer similar insights on MEV opportunities.
Although some opportunities can be in the order of millions USD, the
overwhelming majority of opportunities is below a few hundred USD, or even
less than \$$100$. Also, claiming them requires multiple transfers, which are often
complex and involve multiple smart contracts and accounts.
\ifenhanced
Therefore, a mechanism like the extended transformation
of~\cref{sec:protocol-enhanced}, with a PoW cost parameter in the order of a
few USD, could assist in countering the vast majority of MEV threats and make
the system less costly, and arguably more usable, for regular users.
\fi

\subsection{Usability Considerations}

Introducing delays in transaction generation can adversely impact the system's
usability and impair its ability to respond to external market conditions. For
example, MEV often arises in decentralized exchanges via arbitrage
opportunities, which often exist due to the delayed response
to market movements in centralized exchanges.
Using our mechanism, these arbitrage opportunities could become more
frequent and longer, exacerbating the problem. These delays
could also result in unnecessary liquidations or hinder the users'
ability to manage positions timely.
In essence, although our mechanism is useful in theory, it is not necessarily
suitable for all types of application where MEV opportunities exist.

Intuitively, our mechanism is suitable for preventing front-running in
decentralized applications that do not require frequent and immediate
user responses. For example, NFT generation and registration is such
application, where users typically can afford to wait
before acquiring and/or transferring their tokens. On the other hand, our
mechanism may not be suitable decentralized exchanges and applications
sensitive to price movements, where users need to respond quickly to market
conditions, as discussed in the previous paragraph. Finally, our mechanism is
not suitable for defending against MEV attacks that rely (solely) on
back-running, since its focus is on front-running attacks.

\remark{Our mechanism does not need to be implemented at the ledger protocol
level. Instead, the verifiable delay requirement could be enforced on the smart
contract level, as long as the ledger's smart contract system allows the
verification of delay proofs. This enables some flexibility, since our
mechanism could be employed only on suitable DeFi applications, whereas the
smart contracts of others are implemented via standard practices.}

\ifenhanced
\subsection{Proof-of-Work Cost per Transaction}

\cref{thm:enhanced} shows that, by increasing the cost of producing a
transaction, a protocol can become an equilibrium even when predictable MEV
opportunities exist. As shown above,
opportunities come up every few seconds with a median profit of few tens USD.
Therefore, for the inequality of \cref{thm:enhanced} to hold, the cost $\cost$
of creating a transaction should be also in the order of \$$100$.

Note that such cost would make the system (possibly unreasonably) costly to
use. Also, reaching such cost using PoW is not straightforward in
practice.
For example, for a USA resident, electricity costs range
between $0.0937{-}0.3276$
\$/kWh,\footnote{\url{https://www.electricchoice.com/electricity-prices-by-state}}
so consuming \$$100$ in electricity requires between $300{-}1000$ kWh.
Using a CPU for this is untenable, since a standard Intel CPU
running at $2.0$GHz consumes only $150$ W per hour.
Even using dedicated hardware though is not very practical. For example, a
popular SHA256-based ASIC is Bitmain Antminer S21, which has nominal electricity
consumption $3.5$kW/h.\footnote{\url{https://www.bitmain.com}} Consuming \$$100$
in electricity using a single such machine would require between $85{-}285$
hours, which is clearly not practical for the creation of a single transaction.
\fi

\subsection{Alternative Delay Primitives}

Although we use VDFs, alternative delay-based cryptographic primitives could
also be used. Our application requires a
primitive that is non-parallelizable, provably delayed, and whose solutions
are efficiently verifiable. VDFs guarantee these requirements, but they
also aim to guarantee other properties, such as output uniqueness. In our case,
one could employ alternative primitives, that are more efficient \wrt the
specific properties that we want to achieve. One such possible alternative is
Proof-of-Sequential-Work
(PoSW)~\cite{cryptoeprint:2011/553,cryptoeprint:2018/183}. The main difference
between PoSW and VDFs is that PoSW output is not unique, \st a party can create
multiple valid outputs from a single solution. Although this is not suitable for
other standard applications of delay-based primitives, \eg creating a randomness
beacon, it is not an issue in our case. We leave the exploration of other
delay-based primitives for future work and refer to~\cite{cryptoeprint:2023/687}
as a starting point.

\section{Conclusion}

Our work explores a defense mechanism against Maximal Extractable Value (MEV) in
distributed ledgers, which depends on verifiable delays. Our main idea is
slowing down the ability of a party to respond to a newly-observed transaction,
in a way that a response is only possible after the transaction
is finalized. We present a transformation of a ledger's validity
predicate, which makes use of a Verifiable Delay Function (VDF), and show that
any protocol that is either secure, an equilibrium, or compliant \wrt transaction
censorship when no MEV opportunities exist, can be modified via our
transformation to a protocol that defends against MEV attacks.
\ifenhanced
We then enhance our transformation with Proof-of-Work to defend against predictable MEV
opportunities. 
\fi
Finally, we discuss various implementation details, such as
exploring VDF candidates and reviewing
historical MEV opportunities to evaluate the effectiveness of our mechanism.
Historical analysis indicates that the proposed mechanism can effectively defend
against most MEV cases, albeit only non-predictable ones.

\paragraph{Future work.}
Our work also poses various questions that deem further exploration. Firstly,
although VDFs have been widely discussed and used in many applications, they
lack rigorous theoretical lower bounds on efficiently evaluating them. Until
such bounds are found, mechanisms such as the one presented in this paper rely
only on practical security estimates, \eg based on the most efficient VDF
implementation, instead of concrete theoretical computations.
\ifenhanced
Additionally, the PoW cost estimates of~\cref{sec:implementation} show
that~\cref{thm:enhanced} is mostly theoretical at this point, so future work
could explore increasing a transaction's creation cost in a client-based
manner, that is without requiring complex dedicated hardware. 
\fi
Finally, our
mechanism may be suitable for some applications but not others, \eg that
require frequent and rapid responses to market movements, so a more thorough
exploration of the high-level properties of DeFi applications is needed to
identify in which cases verifiable delays can prevent MEV without incurring
adverse usability effects.

\bibliographystyle{splncs04}
\bibliography{pubs}

\end{document}